\documentclass[12pt,a4paper]{article}
\usepackage{german,amsmath, amsthm, amssymb}
\usepackage{enumerate}
\usepackage{multicol}
\usepackage{epsfig} 
\usepackage{graphicx}
\usepackage[top= 3.0cm, bottom=3.0cm, left=2.5cm, right=2.5cm]{geometry} 
\usepackage[latin1]{inputenc}
\usepackage[T1]{fontenc} 
\usepackage[english]{babel}
\usepackage{dsfont}
\usepackage{cite} 
\usepackage{setspace}
\usepackage{amsthm}
\usepackage{bm}
\newtheorem{theorem}{Theorem}[section]
\newtheorem{proposition}{Proposition}[section]

\theoremstyle{definition}
\newtheorem{definition}{Definition}[section]

\theoremstyle{theorem}
\newtheorem{corr}{Corrollary}[section]

\title{\Huge Essentially Ergodic Behaviour  \\ \huge
}
  \author{\large Paula Reichert\footnote{
 \textit{Mathematisches Institut, Ludwig-Maximilians-Universit\"at M\"unchen, reichert@math.lmu.de}} }

\date{January 17, 2020\\
\normalsize (forthcoming in the \textit{British Journal for the Philosophy of Science})}

\begin{document}

\doublespacing

\maketitle

\begin{abstract}
\normalsize\noindent I prove a theorem on the precise connection of the time and phase space average of the Boltzmann equilibrium showing that the behaviour of a dynamical system with a stationary measure and a dominant equilibrium state is qualitatively ergodic. 
Explicitly, I will show that, given a measure-preserving dynamical system and a region of overwhelming phase space measure, almost all trajectories spend almost all of their time in that region. The other way round, given that almost all trajectories spend almost all of their time in a certain region, that region is of overwhelming phase space measure. In total, the time and phase space average of the equilibrium state approximately coincide. Consequently, equilibrium can equivalently be defined in terms of the time or the phase space average. Even more, since the two averages are almost equal, the behaviour of the system is essentially ergodic. While this does not explain the approach to equilibrium, it provides a means to estimate the fluctuation rates. 
\end{abstract}

\newpage

\section{Introduction}

During the last years, Boltzmann's notion of equilibrium and his explanation of the second law of thermodynamics have been discussed controversially within the philosophical literature. See, for instance, (Lavis [2005], [2008], [2011] and (Frigg and Werndl [2011], [2012a], [2012b], [2015a], [2015b], [2017]). One source of controversy is the \textit{prima facie} discrepancy between Boltzmann's notion of equilibrium, according to which equilibrium corresponds to a macrostate of by far largest phase space measure, and a notion of equilibrium closely connected to thermodynamics\footnote{In standard textbooks on thermodynamics, equilibrium is defined to be a state an isolated system keeps for all times. See, for instance, (Callen [1960]) or (Reiss [1996]).}, according to which equilibrium refers to a state in which an isolated system spends almost all of its time.

What I  will show, in this paper, is that any controversy based on this distinction is unfounded because the two definitions of equilibrium, on the one hand with respect to the phase space average, on the other with respect to the time average, are equivalent. They are equivalent in the sense that the time and phase space average of the equilibrium state approximately coincide. 

Explicitly, what I will show is that, given a dynamical system with a stationary measure and a macro-region of (by far) largest phase space measure, almost all trajectories spend most (almost all) of their time in that region. The other way round, given a region in which almost all trajectories spend most (almost all) of their time, that region is of (by far) largest phase space measure (where the last statement has been shown, differently, by Frigg and Werndl ([2015a], [2015b]).) Noteworthy this result follows from stationarity alone. 

Often ergodicity, or some variant of it, like epsilon-ergodicity (Vranas [1998]; Frigg and Werndl [2011], [2012a]), is introduced to bridge the gap between the time and phase space average of the equilibrium state. What I show is precisely that this appeal to ergodicity is unnecessary because every dynamical system with a stationary measure (that is, in particular, every Hamiltonian system) behaves qualitatively ergodically with respect to the equilibrium state. While, for an ergodic system, the time and phase space average {exactly} coincide, we obtain an {`almost equality'} of the time and phase space average of the equilibrium state. As a consequence, the long-time behaviour of the system is {essentially ergodic}.

I don't claim that this close connection between the time and phase space average of the equilibrium state is something which had to be discovered today. Boltzmann seems to have been well aware of it. This is shown by his estimate on the length of the Poincar\'e cycles (that is, the time a system wanders around phase space before it revisits a tiny region it has started from, see (Boltzmann [1896b]). 
Also, it was certainly clear to everyone who worked in the mathematical field of ergodic theory, like Birkhoff. The reason for why it has never been presented before is presumably twofold. For a mathematician like Birkhoff it is not an interesting result because it follows easily from his theorem (which is indeed hard to prove). Boltzmann, on the other hand, seems to have taken it for granted, once he understood that the equilibrium region is of overwhelming phase space measure, that a typical trajectory would spend almost all of its time in that region, and he didn't need a mathematical theorem to make this more precise. 

In what follows, let me outline the paper. Before I prove the main theorem of the paper, a theorem on the time average of the Boltzmann equilibrium, I will introduce the mathematical setting and, based on that, state the result. This will constitute section 2. 
In section 3, I will introduce Boltzmann's notion of equilibrium. 
Having prepared the grounds for conceiving the result, I will present (section 4.1) and prove (section 4.2) the main theorem of the paper. I will show that, given a region of overwhelming phase space measure and a stationary measure, almost all trajectories spend almost all of their time in that region. In section 4.3, I will prove the converse statement: if almost all trajectories spend almost all of their time in a certain region, that region is of overwhelming phase space measure. Section 5 will finally provide a discussion of the implications and the explanatory value of the results.  


\section{Setup and result}

Let me introduce the mathematical setup.
Let ($\Gamma, \mathcal{B}(\Gamma), \mu$) be a probability space (i.e., $\Gamma$ is a set (what we later call phase space), $\mathcal{B}(\Gamma)$ is the Borel algebra of $\Gamma$ (the set of measurable subsets) and $\mu$ is a probability measure on $\Gamma$, i.e., in particular, $\mu(\Gamma)=1$). Let $T$ be a measure-preserving transformation, that is, for every $A\in\mathcal{B}(\Gamma): \mu(T^{-1}A)=\mu(A)$. This is the same as demanding $\mu$ to be stationary ($T$ is measure-preserving if and only if $\mu$ is a stationary measure). We can now define the time and phase space average of a set $A\in\mathcal{B}(\Gamma)$. 

\begin{definition}[Phase space average]
Let ($\Gamma, \mathcal{B}(\Gamma), \mu$) be a probability space. Let $A\in\mathcal{B}(\Gamma)$. We call 
\begin{equation} \mu(A) = \int_\Gamma \chi_A(x) d\mu(x)\end{equation}
the {`phase space measure'} or {`phase space average'} of $A$.
\end{definition}
Here $\chi_A$ is the characteristic function which attains the value 1 if $x\in A$ and 0 otherwise.

\begin{definition}[Time average]
Let ($\Gamma, \mathcal{B}(\Gamma), \mu$) be a probability space and $T$ a measure-preserving transformation. Let $A\in\mathcal{B}(\Gamma)$. We call
\begin{equation}\hat{A}(x)=  \lim_{\mathcal{T}\to \infty} \frac{1}{\mathcal{T}} \int_{0}^{\mathcal{T}} \chi_{A}(T^tx)dt\end{equation} 
the {`time average'} of $A$ for some $x\in\Gamma$. \end{definition}

I will say more about this limit later. For now it suffices to say that the limit exists pointwise and the limit function is integrable almost everywhere on $\Gamma$.\footnote{For a proof, see (Birkhoff [1931]).} 

While $\mu(A)$ determines the normalized size, or volume, of the region $A$ in phase space, the time average $\hat{A}(x)$ determines the fraction of time (in the infinite-time limit) the trajectory $T^tx$ starting at $x$ spends in the set $A$.

What I will prove is that, for every dynamical system with a measure-preserving transformation, respectively a stationary measure, and a state of overwhelming phase space measure, that is, a phase space average close to one, almost all trajectories spend almost all of their time in that state, that is, its time average is close to one, too. The other way round, given a state in which almost all trajectories spend almost all of their time, that state is of overwhelming phase space measure. Let us, in what follows, refer to states of phase space or time average close to one as `equilibrium states' with respect to the phase space or time average, respectively. What I will show is that a state, which is an equilibrium state with respect to the phase space average, is an equilibrium state with respect to the time average and vice versa. 

Note that the only two assumptions which enter the proof are a) that the measure is stationary and b) that there is a state of overwhelming phase space measure, respectively (for the reverse direction) b)* that there is a state in which almost all trajectories spend almost all of their time. 
Ergodicity doesn't enter the proof, nor do we get ergodicity out of it. However, we get something very similar. While, for an ergodic system, the time and phase space average {exactly} coincide, we obtain an `almost equality' of the time and phase space average for typical initial conditions. Although this result is weaker than ergodicity, it predicts qualitatively the same long-time behaviour of the system.


\section{The Boltzmann equilibrium}

When Ludwig Boltzmann connected the phenomenological (or macro) theory of heat known as thermodynamics to the underlying atomistic (or micro) theory of matter, it was part of his enterprise to account not only for the laws, but also for the concepts of thermodynamics, like the concept of equilibrium. 

According to the standard textbook on thermodynamics, equilibrium is a state in which an isolated system (given that it is there) stays for all times. Respectively, a system is said to be in thermodynamic equilibrium if and only if the macroscopic properties don't change.\footnote{For this definition, see, for instance, (Callen [1960]; Reiss [1996]).} Often it is added that, if the system is not in equilibrium in the beginning, it will evolve into equilibrium very quickly (and then stay there for the rest of its time). All together this is referred to as thermodynamic behaviour. 

While I don't want to present Boltzmann's full account of thermodynamic behaviour, I want to introduce his notion of equilibrium. 
According to Boltzmann, equilibrium refers to a region of by far largest phase space volume, where phase space $\Gamma$ is partitioned into regions of different size by some physical macrovariable or some set of physical macrovariables (thermodynamical variables, like volume $V$, temperature $T$, and so on).\footnote{See (Boltzmann [1896a], [1896b], [1897]). See also (Lebowitz [1993]; Goldstein [2001]; Lazarovici and Reichert [2015]) for a thorough presentation of Boltzmann's account.} Explicitly, every microstate $X$, which is represented by a point on $\Gamma$, determines a certain macrostate $M(X)$, represented by an entire region $\Gamma_M\subset \Gamma$ -- the set of all microstates realizing (respectively, {coarse-graining} to) the given macrostate. 
While a microstate is a particular micro configuration consisting of the exact positions and velocities of all the particles (atoms or molecules), $X=(\textbf{q}_1, ..., \textbf{q}_N, \textbf{p}_1, ..., \textbf{p}_N$), a macrostate $M(X)$ is determined by certain fix values of the macroscopic (thermodynamic) variables of the system, such that two different macrostates are macroscopically distinct. By construction, every macrostate $M(X)$ is realized by a great number of microstates, the set of points constituting $\Gamma_M\subset\Gamma$.

According to Boltzmann, the equilibrium state corresponds to a region $\Gamma_{Eq}\subset \Gamma$ of overwhelming phase space measure.
As such, the following definition is appropriate:
\begin{definition}
[Boltzmann equilibrium] Let $(\Gamma, \mathcal{B}(\Gamma), \mu)$ be a probability space, i.e., in particular, $\mu(\Gamma)=1$. Let $\Gamma$ be partitioned into disjoint, measurable subsets $\Gamma_{M_i} (i=1,...,n)$ by some (set of) physical macrovariable(s), i.e., $\Gamma =\bigcup_{i=1}^n \Gamma_{M_i}$. A set $\Gamma_{Eq}\in\{\Gamma_{M_1}, ..., \Gamma_{M_n}\}$ with phase space average 
\begin{equation}\mu(\Gamma_{Eq})=1-\varepsilon\end{equation} 
where $\varepsilon\in\mathbb{R}$ and $ 0<\varepsilon<<1$ is called the {`equilibrium set'} or {`equilibrium region'}. The corresponding macrostate is called the {`equilibrium state'} of the system.
\end{definition}

Be aware that this definition presumes the existence of a macro-partition. That is, phase space $\Gamma$ is partitioned into regions of different size by some set of physical macrovariables, with different macro-regions corresponding to {macroscopically distinct} states (what we call macrostates). Consequently, it is not an arbitrary value of $\varepsilon$ which, when given, determines an equilibrium region -- such a definition is meaningless from the point of physics. Instead, it is a particular macro-partition, a partition with respect to the physical macrovariables of the theory, which is given and it is with respect to that partition that a region of overwhelming phase space measure, if it exists, defines an equilibrium state in Boltzmann's sense (and, by the way, determines the value of $\varepsilon$). 

In this context, it has been Boltzmann's crucial insight that, for a realistic physical system of about $10^{24}$ particles (where a reasonable number is given by Avogadro's constant) and a partition into {macroscopically distinct} states, there always exists a region of overwhelming phase space measure. This is due to the gap between micro and macro description of the system and the fact that, for a realistic number of particles ($N\approx 10^{24}$), small differences at the macroscopic level translate, at the microscopic level, into huge differences in the phase space volumes, with a proportion of the size of non-equilibrium regions as compared to the equilibrium region not of $1:100$ or $1:1000$, but of the order $1: 10^N$, that is, with $N\approx 10^{24}$, of the order $1: 10^{10^{24}}$ or $1: 10^{1000000000000000000000000}.$\footnote{Be aware that this number, when written explicitly, entirely exceeds the pages of this paper. 
Here the huge difference in phase space volumes is essentially due to the fact that differences in the particular micro distribution of up to $\pm \sqrt{N}$ particles are indistinguishable from a macroscopic point of view, while the set of all  micro configurations which differ from the uniform distribution by up to $\pm \sqrt{N}$ particles constitute the vast majority of all possible configurations (where the latter is a basic result of probability theory; see, for instance, (D\"urr \textit{et al.} [2017]) and the discussion therein).} 

It is this dominance of the equilibrium state which has been noticed and stressed by Boltzmann (for instance, in ([1896b])) and emphasized again later, among others, by Feynman ([1965]), Penrose ([1989], [2004]), Lebowitz ([1993]), and Goldstein ([2001]). Penrose ([1989]) gives a detailed and comprehensive account of the Boltzmannian framework at the example of the gas in a box. There he also computes the volume of the non-equilibrium regions as compared to the volume of the equilibrium region to be about $10^{-N}$ with $N$ being the number of particles involved (that way determining the value of $\varepsilon$). He points out that Boltzmann's framework is general enough to apply to any realistic, macroscopic physical system. Accordingly, there naturally exists a partition into macro-regions of vastly different size with one (equilibrium) region of overwhelming measure. 

In recent philosophical papers, some authors have formulated doubts about the existence of such a state of overwhelming phase space measure (Lavis [2005], [2008], [2011]; Frigg and Werndl [2017]). They claim that while there may be a state of largest measure (what {they} then call the equilibrium state), this need not be a state of overwhelming measure -- it must not even be larger than all other regions together. Of course, this leads to a number of subsequent problems. I think, however, that these doubts are unwarranted. One reason is that, in the examples the authors give, they discern `macrostates' that are \textit{de facto} indistinguishable from a macroscopic point of view (and which as such determine `macro-regions' which differ not too much in size, while macroscopically distinct states determine regions of vastly different volume).\footnote{Of course, the notion of being `macroscopically distinct' is not sharply defined and it need not be in order to see that macroscopically distinct states determine regions in phase space that differ vastly in size. Compare the last footnote and the reference therein.} Hence, it is one important point that the given `macro-partitions' are not relevant to physics, to say the least.

It shall, however, not be the aim of this paper to discuss these matters, which have already been addressed  elsewhere (see, for instance, (Lazarovici and Reichert [2015]) or, more recently, (Lazarovici [2018]) where the above-mentioned objection is clarified and discussed in detail). 
In what follows, let me therefore assume that the value of $\varepsilon$ which has been proposed by Penrose and which is in accord with Boltzmann's reasoning, $\varepsilon\approx 10^{-N}$ with $N\approx 10^{24}$, provides just the right oder of magnitude relevant to physics.\footnote{Be aware that the theorem I am going to prove in the next section holds for any value of $\varepsilon$ between zero and one. Note, in addition, that a sensible notion of typicality can already be established for values of $\varepsilon$ much larger than $10^{-24}$. All this notwithstanding, it is this value of $\varepsilon$ ($\varepsilon\approx 10^{-24}$) which is provided by the analysis of realistic, normal-sized physical systems. Hence, when we later evaluate the theorem, we assume that this is the value  of $\varepsilon$ which determines the assertions relevant to physics.}

\section{The time average of the Boltzmann equilibrium}

In this section, which is rather technical, I will prove the main theorem of this paper, a theorem on the time average of the Boltzmann equilibrium. 

\subsection{A theorem on the time average}

Consider a macro-partition and within that partition a state of overwhelming phase space measure, i.e., an equilibrium state in the sense of Boltzmann. Thus there is an equilibrium region $\Gamma_{Eq}$ with measure $\mu(\Gamma_{Eq})=1-\varepsilon$ where $0<\varepsilon<<1$ (where the precise value of $\varepsilon$ depends on the particular model and macro-partition, but where a reasonable value of $\varepsilon$ is taken to be of the order $10^{-N}$ with $N\approx 10^{24}$). In fact, the theorem below holds for any value of $\varepsilon$ between zero and one, but the result gets more pronounced the smaller $\varepsilon$. 

In what follows, I will show that the set of trajectories which spend a fraction of less than $1-\sqrt{\varepsilon}$ of their time in the equilibrium region $\Gamma_{Eq}$ is very small (of measure less than $\sqrt{\varepsilon}$). Respectively, the other way round, the set of trajectories which spend at least $1-\sqrt{\varepsilon}$ of their time in equilibrium is very large (of measure greater than $1-\sqrt{\varepsilon}$). Moreover, the smaller the mean time average of a set of trajectories, the smaller the bound on the measure of the respective set. 

Let, again, ($\Gamma, \mathcal{B}(\Gamma), \mu$) be a probability space, i.e., in particular, $\mu(\Gamma)=1$. Let, in addition, $T$ be a measure-preserving transformation, i.e., for every $A\in\mathcal{B}(\Gamma): \mu(T^{-1}A)=\mu(A)$. 
In what follows, we will consider the set $Z$ of points $x\in\Gamma$ for which the time average of equilibrium is smaller than $1-k\varepsilon$: $Z=\{x\in\Gamma|\hat{\Gamma}_{Eq}(x)<1-k\varepsilon\}$ with $1\le k\le1/\varepsilon$. These points determine trajectories which spend a fraction of less than $1-k\varepsilon$ of their time in equilibrium. We will then compare it to the set $X$ of points which spend a fraction of at least $1-k\varepsilon$ of their time in equilibrium: $X=\{x\in\Gamma|\hat{\Gamma}_{Eq}(x)\ge1-k\varepsilon\}$. With respect to these two sets we can state the following theorem and corollary.

\begin{theorem}[Time average of $\Gamma_{Eq}$]
Let ($\Gamma, \mathcal{B}(\Gamma), \mu$) be a probability space and let $T$ be a measure-preserving transformation. Let $\varepsilon, k\in\mathbb{R}$ with $0<\varepsilon<<1$ and $1\le k\le 1/\varepsilon$. Let $\Gamma_{Eq}\subset\Gamma$ be an equilibrium region, i.e., in particular $\mu(\Gamma_{Eq})=1-\varepsilon$. 
Let $Z$ be the set of points for which the time average of the equilibrium region is smaller than $1-k\varepsilon$: $Z=\{x\in\Gamma|\hat{\Gamma}_{Eq}(x)<1-k\varepsilon\}$. 
It follows that $Z$ is of measure
\begin{equation} \mu(Z) <1/k. \end{equation}

\end{theorem}
 
From this you get a bound on the phase space measure of the set $X$: 

\begin{corr}[]
Let everything be as in the above theorem. Let $X$ be the set of points for which the time average of the equilibrium region is larger than or equal to $1-k\varepsilon$: $X=\{x\in\Gamma|\hat{\Gamma}_{Eq}(x)\ge1-k\varepsilon\}$. Then $X$ is of measure \begin{equation} \mu(X)> 1- 1/k.\end{equation}

\end{corr} 
 
I will prove this theorem below (in section 4.2).
In fact, this theorem and corollary hold for any value of $\varepsilon$ and $k$ with $1\le k\le 1/\varepsilon$. We are, however, particularly interested in the case in which $\varepsilon$ is very small (of the order $10^{-N}$ with $N\approx 10^{24}$). In that case, we can choose $k$ within the given bounds $1\le k\le 1/\varepsilon$ large enough for $\mu(Z)$ to be close to zero and $\mu(X)$ to be close to one. A convenient choice of $k$ is the following: \[k=1/\sqrt{\varepsilon}.\] 
In that case, we distinguish between the `good' set $X$ of trajectories which spend at least $1-\sqrt{\varepsilon}$ of their time in equilibrium and the `bad' set $Z$ of trajectories which spend less than $1-\sqrt{\varepsilon}$ of their time in equilibrium. From the above theorem and corollary we obtain that the phase space measure of the `bad' set $Z$ is less than $ \sqrt{\varepsilon}$ whereas the phase space measure of the `good' set $X$ is larger than $1-\sqrt{\varepsilon}$:
\[ \mu(Z)<\sqrt{\varepsilon} \hspace{1cm} \mathrm{and} \hspace{1cm}Ê\mu(X)>1-\sqrt{\varepsilon}.\]

Inserting the value of $\varepsilon$ from above, $\varepsilon\approx 10^{-10^{24}}$, we can restate the result in its   relevant form. In that case, the equilibrium region $\Gamma_{Eq}$ is of measure $\mu(\Gamma_{Eq})=1-10^{-10^{24}}$. Moreover, from $\varepsilon\approx 10^{-10^{24}}$ it follows that $\sqrt{\varepsilon}\approx 10^{-10^{23}}$. Thus we obtain that the phase space measure of the set $Z$ of trajectories which spend less than $1-10^{-10^{23}}$ of their time in equilibrium is vanishingly small (close to zero), whereas the phase space measure of the set $X$ of trajectories which spend at least $1-10^{-10^{23}}$ of their time in equilibrium is overwhelmingly large (close to one). To be precise: 
\[ \mu(Z)<10^{-10^{23}} \hspace{1cm} \mathrm{and} \hspace{1cm}Ê\mu(X)>1-10^{-10^{23}}.\]

Using the notion of typicality as presented, e.g., in D\"urr et al. (2017), we thus find that trajectories which spend almost all of their time in equilibrium are {typical} whereas trajectories which spend less time in equilibrium are {atypical}!

Note that to derive this result, only two assumptions are needed, one of which we discussed above (in section 3):

\begin{itemize}

\item {stationarity:} the phase space measure $\mu$ on $\Gamma$ is stationary, i.e., for any measurable set $A\subset\Gamma$: $\mu(A)=\mu(T^{-1}A$), and
\item {dominance of the equilibrium state:} the equilibrium state determines a region $\Gamma_{eq}$ of by far largest phase space measure, i.e., $\mu(\Gamma_{Eq})=1-\varepsilon$ with $0<\varepsilon<<1$.

\end{itemize}
We take it that both are natural conditions in case we consider a realistic, macroscopic physical system. While the second condition has been discussed in section 3, the first condition, stationarity of the measure, is a weak assumption about the dynamics. It is weak in the sense that every Hamiltonian system comes equipped with a stationary measure, the Liouville measure. Based on that, there exist a couple of stationary measures, like the microcanonical measure, the canonical measure, and so on, all of them frequently used in statistical mechanics. In short, stationarity is provided by the dynamics.

\subsection{Proof of the theorem on the time average}

Let me now prove the theorem on the time average of $\Gamma_{Eq}$.

\begin{proof}[Proof (Theorem 4.1)]
The transformation $T$ is measure-preserving, that is, for any set $A\in \mathcal{B}(\Gamma)$ and $\forall t$: $\mu(A)=\mu(T^{-t}A)$. Hence, in particular, $\mu(\Gamma_{Eq})=\mu(T^{-t}\Gamma_{Eq})$ where $\Gamma_{eq}$ refers to the equilibrium state, i.e., $\mu(\Gamma_{Eq})=1-\varepsilon$. It follows that $\mu(T^{-t}\Gamma_{Eq})=1-\varepsilon$, as well, and, hence,
\begin{eqnarray}1-\varepsilon &=&\mu(\Gamma_{Eq})= \mu(T^{-t}\Gamma_{Eq}) = \int_\Gamma \chi_{T^{-t}\Gamma_{Eq}}(x)d\mu(x)=\int_\Gamma \chi_{\Gamma_{Eq}}(T^tx)d\mu(x)\nonumber\\
&=& \lim_{\mathcal{T}\to \infty} \frac{1}{\mathcal{T}} \int_{0}^{\mathcal{T}} dt \int_\Gamma \chi_{\Gamma_{Eq}}(T^tx)d\mu(x).\end{eqnarray}
Here the last equation follows from the fact that the integrand is a constant.

At this point we make use of the theorem of Birkhoff ([1931])\footnote{For a thorough presentation of Birkhoff's theorem, see also (Petersen [1983]).} which says that, for any measure-preserving transformation $T$ 
and for any $\mu$-integrable function $f$, i.e. $f\in L^1(\mu)$, the limit 
\[ \hat{f}=\lim_{\mathcal{T}\to \infty} \frac{1}{\mathcal{T}} \int_{0}^{\mathcal{T}} f(T^tx)dt\]
exists for almost every $x\in\Gamma$ and the (almost everywhere defined) limit function $\hat{f}$ is integrable, i.e., $\hat{f}\in L^1(\mu)$. 

Let us apply Birkhoff's theorem to the above equation. The characteristic function $\chi_{\Gamma_{Eq}}$ is $\mu$-integrable and, hence, for almost all $x\in \Gamma$, the limit $\lim_{\mathcal{T}\to \infty} \frac{1}{\mathcal{T}} \int_{0}^{\mathcal{T}} \chi_{\Gamma_{Eq}}(T^tx)dt$ exists and is $\mu$-integrable. In other words, for almost every single trajectory the time average exists. In that case, we can change the order of integration and, by dominated convergence, pull the limit into the $\mu$-integral. Let $\Gamma^*\subset\Gamma$ with $\mu(\Gamma^*)=\mu(\Gamma)$ be the set of points for which the time average exists. Then equation (6) becomes
\begin{equation} 1-\varepsilon=\int_{\Gamma^*} d\mu(x) \bigg[\lim_{\mathcal{T}\to \infty} \frac{1}{\mathcal{T}} \int_{0}^{\mathcal{T}} \chi_{\Gamma_{Eq}}(T^tx)dt\bigg].\end{equation}

For means of demonstration, let me show how this equation is fulfilled in the two `extreme' cases, when the dynamics is very special: first, the ergodic case, where the trajectory is dense in phase space and, second, the case in which $T^t$ is the identity, where the trajectory is merely one point. All other cases lie in between. 

The first way to fulfill equation (7) is the following: the time average $\hat{\Gamma}_{Eq}(x)$ is a constant (almost everywhere). In that case, it must hold that $\hat{\Gamma}_{Eq}(x)=1-\varepsilon$. The set of all points $x\in \Gamma$ for which the limit exists (and is constantly $1-\varepsilon$), defines an invariant set, $T^{-1}\Gamma^*=\Gamma^*,$ with measure $\mu(\Gamma^*)=1$. This is the ergodic case. 

The second way to fulfill equation (7) is that there exists an invariant set $A$ (i.e. $T^{-1}A=A$) with $\mu(A)=\varepsilon$ such that $\forall x\in A$: $ \hat{\Gamma}_{Eq}(x)=0$ and $\forall x \notin A:  \hat{\Gamma}_{Eq}(x)=1$ (again, up to a set of measure zero). Then also $\Gamma^*\backslash A$ is an invariant set and $\mu(\Gamma^*\backslash A)=1-\varepsilon$. This reflects the case of $T^t$ being the identity, $T^tx=x$, and $\Gamma\backslash A=\Gamma_{Eq}$. 

Let us now analyse the general case. Let again $X$ be the set of all $x\in\Gamma$ which share the property that their time average is at least $1-k\varepsilon$ with $1\le k\le 1/\varepsilon$. That is, 
$X=\{x\in \Gamma| \hat{\Gamma}_{Eq}(x)\ge 1-k\varepsilon\}.$
In contrast, let again $Z$ be the set of all $x\in\Gamma$ for which the time average is smaller than $1-k\varepsilon$: $Z=\{x\in \Gamma| \hat{\Gamma}_{Eq}(x)< 1-k\varepsilon\}.$
It is clear that such a decomposition exists, that is, $\Gamma^*=X\cup Z$ ($\Gamma=X\cup Z$ up to a set of measure zero), $\mu(Z)=\mu(\Gamma\backslash X)$, and $X$ and $Z$ are invariant sets.  
Using the definition of $X$ and $Z$, equation (7) can be rewritten as 
\begin{equation}1-\varepsilon
= \int_{ X} \hat{\Gamma}_{Eq}(x) d\mu(x) + \int_{ Z} \hat{\Gamma}_{Eq}(x) d\mu(x).
\end{equation}

Let now the {`mean time average'} of $X$ be defined as \[\bar{\Gamma}_{eq}(X)= \frac{1}{\mu(X)}\int_X \hat{\Gamma}_{Eq}(x)d\mu(x),\]where $\hat{\Gamma}_{Eq}(x)$ exists and is integrable for all $x\in X$ (this is part of the definition of $X$). The mean time average determines the mean fraction of time the trajectories starting in $X$ spend in the set $\Gamma_{Eq}$. Analogously, let $\bar{\Gamma}_{eq}(Z)$ denote the mean time average of  $Z$. 
Using the definition of the mean time average, equation (8) can be rewritten as
\begin{equation}Ê1-\varepsilon= \bar{\Gamma}_{Eq}(X)\mu(X) + \bar{\Gamma}_{Eq}(Z)\mu(Z). \end{equation}
We want to solve this for $\mu(Z)$. Note that $\mu(X)=1-\mu(\Gamma\backslash X)=1-\mu(Z)$. 
Moreover, since $\lim_{T\to \infty} \frac{1}{T} \int_{0}^{T} \chi_{\Gamma_{Eq}}(T^tx)dt\le1$: $\bar{\Gamma}_{Eq}(X)\le1$. 
On the other hand, it follows from the definition of the mean time average that $\bar{\Gamma}_{Eq}(Z)<1-k\varepsilon$ (since $\hat{\Gamma}_{Eq}(x)<1-k\varepsilon$ for all $x\in Z$). Hence, since $k\ge 1$: $\bar{\Gamma}_{Eq}(Z)<1-\varepsilon$.

Now in order for the right hand side of equation (9) to add up to $1- \varepsilon$, the measure of $Z$ needs to be small. This is due to the fact that $\mu(Z)$ comes with a factor $\bar{\Gamma}_{Eq}(Z)<1-\varepsilon$ which can only be encountered by a factor $\bar{\Gamma}_{Eq}(X)\ge 1-\varepsilon$ in front of $\mu(X)$. However, since $\bar{\Gamma}_{Eq}(X)$ is bounded from above by one, $\bar{\Gamma}_{Eq}(X)\le 1$, the first summand can outweigh the second only if $\mu(X)$ is large enough (respectively, $\mu(Z)$ small enough). At most, $\bar{\Gamma}_{Eq}(X)=1$. In that case, $\mu(X)$ attains its minimum and $\mu(Z)$ its maximum (where $\mu(Z) = 1- \mu(X)$). 
Since we want to determine an upper bound on $\mu(Z)$, we set $\bar{\Gamma}_{Eq}(X)=1$ (a condition we will relax later). Let, in addition, $\Theta:=\mu(Z)$. Then equation (9) can be rewritten as
\begin{equation} 1-\varepsilon = (1-\Theta) + \bar{\Gamma}_{Eq}(Z)\Theta \end{equation} 
With $\bar{\Gamma}_{Eq}(Z)<1 -k\epsilon$, it follows that
\begin{equation} \hspace{0.5cm}\Theta=\frac{\varepsilon}{1-  \bar{\Gamma}_{Eq}(Z)} < \frac{\varepsilon}{1- (1-k\varepsilon)} =1/k. \end{equation}
If we now no longer restrict the mean time average of $X$ to be one, this inequality becomes even more pronounced. That way we obtain an upper bound on $\mu(Z)$:
\[\mu(Z)< 1/k.\] 
This proves the assertion.
\end{proof}

It is now easy to prove {corollary 4.1}:

\begin{proof}[Proof (Corollary 4.1)]
Given the fact that $\mu(Z)<1/k$, it follows that 
\[ \mu(X)=\mu(\Gamma\backslash Z) = 1-\mu(Z)> 1-1/k.\]
\end{proof}

\subsection{The converse statement}

The converse statement follows directly from the almost-everywhere existence and integrability of the time average. It says that if there exists a region $\Gamma_M\subset \Gamma$ such that almost all trajectories (all $x\in X$ with $\mu(X)=1-\delta$) spend almost all of their time in that region ($\forall x\in X$: $\hat{\Gamma}_M(x)\ge 1-\varepsilon$), then this region has large phase space measure: $\mu(\Gamma_{M})\ge (1-\varepsilon)(1-\delta)$. Here $\delta$ and $\varepsilon$ are assumed to be small, $0<\delta<<1$ and $0<\varepsilon<<1$ (this is the case we are interested in, while, in fact, the proof holds for any value of $\delta$ and $\varepsilon$ between 0 and 1). The following proposition can be proven.
\begin{proposition}
Let the setting be as above. Let $\Gamma_M\subset \Gamma$ and $X\subset \Gamma$ with $\mu(X)=1-\delta$ be such that $\forall x \in X$: $\hat{\Gamma}_{M}(x)\ge1-\varepsilon$. 
Then 
\begin{equation} \mu(\Gamma_M) \ge (1-\varepsilon)(1-\delta).\end{equation}
\end{proposition}

\begin{proof}
When you apply equations (7) and (8) to the set $\Gamma_M\subset\Gamma$, you get 
\[ \mu(\Gamma_{M})= \int_{\Gamma^*} d\mu(x) \bigg[\lim_{\mathcal{T}\to \infty} \frac{1}{\mathcal{T}} \int_{0}^{\mathcal{T}} \chi_{\Gamma_{M}}(T^tx)dt\bigg]=\int_X d\mu(x) \hat{\Gamma}_{M}(x)+ \int_Z d\mu(x)\hat{\Gamma}_{M}(x), \]
where the first equality holds due to Birkhoff's theorem and the second uses the definition of the time average and the fact that $Z=\Gamma^*\backslash X$. From $ \int_Z d\mu(x)\hat{\Gamma}_{M}(x)\ge 0$ and the assumptions it follows that 
\[ \mu(\Gamma_{M})\ge\int_X d\mu(x) \hat{\Gamma}_{M}(x)\ge (1-\varepsilon)(1-\delta). \]
This proves the assertion.
\end{proof}

Let $0<\delta<<1$ and $0<\varepsilon<<1$. Then this result tells us that if there exists a region $\Gamma_M\subset \Gamma$ such that almost all trajectories (all $x\in X$ with $\mu(X)=1-\delta$) spend almost all of their time ($\forall x\in X$: $\hat{\Gamma}_M(x)\ge 1-\varepsilon$) in that region, this region is of overwhelming phase space measure: $\mu(\Gamma_{M})\ge (1-\varepsilon)(1-\delta)$. Frigg and Werndl ([2015a], [2015b]) show qualitatively the same. They do this by means of the ergodic decomposition of the system. We instead obtain the result directly from the almost-everywhere existence and integrability of the time average.

What is the physical meaning of the state $\Gamma_M$?
Frigg and Werndl ([2015a], [2015b]) coin the notion of an `$\alpha$-equilibrium' referring to a region in which typical trajectories spend most of their time (to be precise, a fraction of $\alpha>1/2$). This is essentially what we call $\Gamma_M$ (with $\varepsilon<1/2$). According to them, it is the notion of an $\alpha$-equilibrium (where equilibrium is defined in terms of the time average) rather than Boltzmann's notion of equilibrium (where equilibrium is defined in terms of the phase space average) which connects to the thermodynamic notion of equilibrium. 

This is true given that equilibrium in thermodynamics is defined with respect to the long-time behaviour of the system -- equilibrium is the state in which the system stays for all times. The problem, however, is that the notion of an $\alpha$-equilibrium is empty as long as it doesn't refer to a physical state, just like, in the case of Boltzmann, an arbitrary value of $\varepsilon$ doesn't in itself define an equilibrium region.\footnote{From a purely mathematical point of view there exist uncountably many $\alpha$-equilibriae. For example, phase space $\Gamma$ is itself an $\alpha$-equilibrium: all trajectories spend all of their time in $\Gamma$. However, such regions don't \textit{per se} connect to physical states. Also note that there is no way to identify a particular $\alpha$-equilibrium by direct observation since we cannot wait infinitely long (as demanded by the infinite-time limit). Hence, the notion of an $\alpha$-equilibrium is empty in itself and becomes meaningful only in connection with Boltzmann's notion.} 

If we want to recover thermodynamics, we have to introduce a set of macrovariables (thermodynamical variables, like volume $V$, temperature $T$, etc.) in order to, first of all, specify the macroscopic physical state of the system. However, once we specified a macrostate, we are back in the previous, Boltzmannian setting, where there is a supervenience of macrostates over microstates, with every macrostate corresponding to a particular, measurable region of phase space. Only that now, by help of theorem 4.1, we can indeed determine the time average of the equilibrium region, at least for typical trajectories, and check, by that means, whether it is an $\alpha$-equilibrium. 

\section{Discussion}
From theorem 4.1 we learned that, given a system with a stationary measure and an equilibrium region of by far largest measure, typical initial data determine trajectories which spend almost all of their time in that region. 
To be precise, a set of points of measure larger than $1-\sqrt{\varepsilon}$ determines trajectories that spend a fraction of at least $1-\sqrt{\varepsilon}$ of their time in equilibrium. Respectively, the other way round, the set of trajectories which spend a fraction of less than $1-\sqrt{\varepsilon}$ of their time in equilibrium is of measure less than $\sqrt{\varepsilon}$. Moreover, the smaller the mean time average, the smaller the bound on the respective set.

Now consider again a reasonable value of $\varepsilon$, that is, $\varepsilon \approx10^{-N}$, where $N\approx 10^{24}$ for a realistic, macroscopic physical system. For this value of $\varepsilon$, the set of trajectories which spend almost all of their time in equilibrium ($\ge 1-10^{-10^{23}}$) is of measure close to one ($>1-10^{-10^{23}}$), while the measure of the set of trajectories which spend little time in equilibrium ($<1-10^{-10^{23}}$) is close to zero ($<10^{-10^{23}}$). In other words, the trajectories which spend almost all of their time in equilibrium are {typical}, whereas the trajectories which spend little time in equilibrium are {atypical}.

What does this imply for the individual system?\footnote{For the connection between typical properties, typical sets and typical individuals as well as for the status of typicality explanations in general, see (Lazarovici and Reichert [2015]) or, more recently, (Wilhelm [2019]).} Consider an isolated system at an arbitrary moment of time. What we obtain from the above result is that typically the system is in equilibrium {at that moment} and it has been and will be in equilibrium {for almost all past and future times}, thus exhibiting, essentially, what Lavis ([2005]) calls thermodynamic-like behaviour (where `thermodynamic-like' in contrast to `thermodynamic' allows for (rare) fluctuations out of equilibrium). 

This is precisely how Boltzmann's notion of equilibrium connects to the thermodynamic notion of equilibrium. And this is how we encounter an objection which has been formulated by Frigg and Werndl in ([2015a], [2015b]). They claim  that Boltzmann's notion of equilibrium or, more generally,
\begin{quote}
[Boltzmann's account] faces a serious problem: the absence of a conceptual connection with the thermodynamic (TD) notion of equilibrium [where the] following is a typical TD textbook definition of equilibrium: `A thermodynamic system is in equilibrium when none of its thermodynamic properties are changing with time [. . .]'. (Frigg and Werndl [2015a], p. 12)
\end{quote}
 The above result shows that this statement is simply false. There is not only a conceptual, but even a definite mathematical connection. Given an equilibrium state of overwhelming measure (an equilibrium state \`a la Boltzmann), the typical system spends almost all of its time in that state, thus exhibiting thermodynamic behaviour (non-changing macro properties) up to (rare) fluctuations. 

At this point, it is understood that fluctuations, which, in total, add up to a tiny fraction of time, constitute a small refinement to the thermodynamic notion of equilibrium. 
On observational scale, thermodynamics and Boltzmann coincide. In fact, the above result asserts that the typical system is in equilibrium {basically {all} the time} (a fraction of time $\ge 1-10^{-10^{23}}$). Accordingly, we don't expect to observe {any} fluctuation {at all}, neither on the time scale relevant to experiments, nor on the cosmological scale given by the age of our universe (compare the estimate on the fluctuation rates at the end of this discussion).

While the given result connects to the thermodynamic notion of equilibrium, it {does not} explain the irreversible (time-asymmetric) part of thermodynamic behaviour which is usually formulated as an amendment to any textbook definition of thermodynamic equilibrium, namely that if the system is {not} in equilibrium in the very beginning, it will evolve into equilibrium rather quickly. This is not shown, and cannot be shown by the above result, first and foremost, due to the fact that there is a $t\to \infty$ limit in the definition of the time average. Thus, just like in case of an ergodic (or epsilon-ergodic) system -- ergodicity (or epsilon-ergodicity) includes the very same $t\to \infty$ limit! --, also in this case we cannot infer anything about the behaviour of the system on small time scales, like the scale relevant to the system's approach to equilibrium. To put it differently, the $t\to \infty$ limit allows anything on small time scales and there is, at this level of description, no argument against, let's say, lengthy fluctuations of several hours (or days or weeks) which still add up to a tiny fraction of time with respect to the infinite-time limit.\footnote{See also (Bricmont [1995]) for a thorough discussion of the notion of ergodicity (which also applies to epsilon-ergodicity) and its failure to contribute to the explanation of the approach to equilibrium within short times. Note that the same arguments apply to epsilon-ergodicity (a concept proposed by Frigg and Werndl in ([2011], [2012a]), which includes the very same infinite-time limit.} 

If we want to explain the rapid evolution from non-equilibrium to equilibrium, we have to return to a refined description of the system's motion on phase space, where there are macro-regions of vastly different size (with one region of overwhelming volume) and there is a trajectory winding around erratically due to the chaotic behaviour of the system. It is essentially this erratic behaviour which, together with the huge size of the equilibrium region, explains the rapid approach to equilibrium via an explanatory scheme used by Boltzmann and known today as the typicality account. Since it is not the aim of this paper to discuss the typicality account, let me, at this point, just refer to the pertinent literature on this topic. See, for instance, (Lebowitz [1993]; Bricmont [1995]; Goldstein [2001]).\footnote{See also (Frigg and Werndl [2012b], [2015a], [2015b]) for a critical discussion of the typicality account and (Lazarovici and Reichert [2015]; Lazarovici [2018]) for responses to that.}

In what follows, take into account the reverse statement, given in proposition 4.1, which asserts that any state in which almost all trajectories spend most (almost all) of their time is of (by far) largest phase space measure. Together theorem 4.1 and proposition 4.1 provide the two directions of proof which show that the time and phase space average of the equilibrium state, be it an equilibrium with respect to the time or the phase space average, respectively, approximately coincide. In other words, if a state is of phase space average close to one, its time average is (typically) close to one and vice versa.

To have an idea of the closeness of the two averages, consider again a reasonable value of $\varepsilon$, like, e.g., $\varepsilon= 10^{-10^{24}}$. Let, for means of simplicity, also $\delta= 10^{-10^{24}}$. As before, $\sqrt\varepsilon= 10^{-10^{23}}$. For the first direction, reconsider theorem 4.1 with $k=1/\sqrt{\varepsilon}$.  In that case, the equilibrium region is of measure $\mu(\Gamma_{Eq})=1-\varepsilon$ and we obtain from theorem 4.1 that, on a set $X$ of measure close to one, the time average is $\hat\Gamma_{Eq}(x)\ge1-\sqrt{\varepsilon}$. For $x\in X$ and the given value of $\varepsilon$, one obtains:
\[ \mu(\Gamma_{Eq})= 1-10^{-10^{24}} \hspace{1cm} \Rightarrow \hspace{1cm}\hat{\Gamma}_{Eq}(x)\ge 1-10^{-10^{23}}.\]
Hence, the two averages differ at most by the order of $\sqrt{\varepsilon}= 10^{-10^{23}}$. What regards the converse statement, we obtain from proposition 4.1 that if, on a set $X$ of measure close to one, the time average of a certain region is $\hat\Gamma_{M}(x)\ge 1-\varepsilon$, that region is of measure $\mu(\Gamma_{M})\ge (1-\varepsilon)^2> 1-2 \varepsilon$. Hence, for $x\in X$ and the given values of $\delta$ and $\varepsilon$, it is conversely:
\[ \mu(\Gamma_{M})> 1-2\cdot 10^{-10^{24}} \hspace{1cm} \Leftarrow \hspace{1cm}\hat\Gamma_{M}(x)\ge 1- 10^{-10^{24}} .\] 
In this case, the two averages differ at most by the order of ${\varepsilon}=10^{-10^{24}}$. All in all, the two averages are equal up to the order of $\sqrt\varepsilon= 10^{-10^{23}}$!
This has important implications. 

First, it shows that one should not highlight the difference between Boltzmann's definition of equilibrium, which refers to the phase space average, and a definition with respect to the time average, like the one Frigg and Werndl ([2015a], [2015b]) propose. Both definitions are equivalent in the sense that, if a state is an equilibrium state with respect to the phase space average, it is an equilibrium state with respect to the time average and vice versa. 

Second, we obtain from the `almost equality' of the two averages that the system behaves essentially as if it were ergodic. While, for an ergodic system, the time and phase space average of the equilibrium state exactly coincide, the approximate equality of the two averages implies qualitatively the same behaviour: given a region of overwhelming measure, almost all trajectories spend almost all of their time in that region; the other way round, given a region in which almost all trajectories spend almost all of their time, that region is of overwhelming measure.

As a consequence, there is simply no need to refer to ergodicity in discussions on the foundations of statistical mechanics, as it has often been done at the time when ergodic theory was at its height (see the references in (Bricmont [1995])), but also recently (see, for instance, (Frigg and Werndl [2011], [2012a])). The preceding result predicts qualitatively the same behaviour, what we call essentially ergodic behaviour, for any system with a stationary measure and a dominant equilibrium state.

Third, the given result justifies Boltzmann's assumption of ergodicity in his computation of the recurrence times (see, for instance, Boltzmann's letter to Zermelo ([1896b])). In that letter, Boltzmann uses ergodicity to calculate the length of the Poincar\'e cycles, respectively the recurrence times, i.e., the time a system needs to wander around phase space before it revisits a tiny region it has started from. Since any system with a stationary measure and a dominant equilibrium state behaves essentially ergodically, Boltzmann's assumption of ergodicity is a legitimate tool in such approximate computations.

Recall that ergodicity (or any related concept, like the essential ergodicity we deal with here) is relevant to the behaviour of the system only on time scales connected to the Poincar\'e cycles, when the trajectory has been winding around all of  phase space covering it more or less densely. This is the time scale at which ergodicity begins to matter. Consequently, the given result doesn't explain the approach from non-equilibrium to equilibrium (which happens within short times), but it provides an estimate on the rate of fluctuations, i.e., the time a system spends in equilibrium until it revisits a tiny non-equilibrium region. 

What then is a good estimate of the rate of fluctuations? To see this, reconsider the theorem on the time average of the Boltzmann equilibrium and, particularly, the set $X$ of trajectories which, if the equilibrium region is of measure $1-10^{-10^{24}}$, spend at least $1-10^{-10^{23}}$ of their time in equilibrium. Recall that the set $X$ if of measure $\mu(X)> 1-10^{-10^{23}}$, that is, it consists of almost all trajectories. Since these trajectories spend at least $1-10^{-10^{23}}$ of their time in equilibrium, they spend less than $10^{-10^{23}}$ of their time textit{out of} equilibrium, that is, in a fluctuation. 
This gives us an idea of the long-time behaviour of the system. If we assume that fluctuations happen occasionally, in accordance with a trajectory wandering around phase space erratically, then we can estimate the approximate rate of fluctuations, for a typical trajectory, as follows: a fluctuation of $1$ second occurs about every $10^{10^{23}}$ seconds, that is, about every $ 10^{10^{23}}$ years. But this means that the system spends trillions of years in equilibrium as compared to one second in non-equilibrium, a time larger than the age of the universe!

\newpage

\Large \noindent \textbf{References}
\normalsize
\vspace{0.8cm}

\noindent Birkhoff, G. D. [1931]: `Proof of the Ergodic Theorem', \textit{Proceedings of the National Academy of Sciences of the United States of America}, \textbf{17}, pp. 65--60.\vspace{0.5cm}

\noindent Boltzmann, L. [1896a]: \textit{Vorlesungen \"uber Gastheorie}, Leipzig: Verlag v. J. A. Barth. Nabu Public Domain Reprints.\vspace{0.5cm}

\noindent Boltzmann, L. [1896b]: `Entgegnung auf die w\"armetheoretischen Betrachtungen des Hrn. E. Zermelo', \textit{Wiedemann's Annalen}, \textbf{57}, pp. 773--84.\vspace{0.5cm}

\noindent  Boltzmann, L. [1897]: `Zu Hrn. Zermelos Abhandlung ``\"Uber die mechanische Erkl\"arung irreversibler Vorg\"ange''', \textit{Wiedemann's Annalen}, \textbf{60}, pp. 392--8. \vspace{0.5cm}

\noindent  Bricmont, J. [1995]: `Science of Chaos or Chaos in Science?', \textit{Physicalia Magazine}, \textbf{17}, pp. 159--208. \vspace{0.5cm}

\noindent Callen, H. B. [1960]: \textit{Thermodynamics and an Introduction to Thermostatics}, New York: Wiley.\vspace{0.5cm} 

\noindent  D\"urr, D., Fr\"omel, A., and Kolb, M. [2017]: \textit{Einf\"uhrung in die Wahrscheinlichkeitstheorie als Theorie der Typizit\"at}, Berlin: Springer. \vspace{0.5cm}

\noindent Frigg, R. and Werndl, C. [2011]: `Explaining Thermodynamic-Like Behavior in Terms of Epsilon-Ergodicity', \textit{Philosophy of Science}, \textbf{78}, pp. 628--52.  \vspace{0.5cm}

\noindent Frigg, R. and Werndl, C. [2012a]: `A New Approach to the Approach to Equilibrium', in Y. Ben-Menahem and M. Hemmo (eds), \textit{Probability in Physics}, Berlin: Springer, pp. 99--113.\vspace{0.5cm}

\noindent Frigg, R. and Werndl, C. [2012b]: `Demystifying Typicality', \textit{Philosophy of Science}, \textbf{79}, 917--929. \vspace{0.5cm}

\noindent Frigg, R. and Werndl, C. [2015a.] `Reconceptualising Equilibrium in Boltzmannian Statistical Mechanics and Characterising its Existence', \textit{Studies in History and Philosophy of Modern Physics}, \textbf{49:1}, pp. 19--31. \vspace{0.5cm}

\noindent  Frigg, R. and Werndl, C. [2015b]: `Rethinking Boltzmannian Equilibrium',  \textit{Philosophy of Science} \textbf{82:5}, pp. 1224--35. \vspace{0.5cm} 

\noindent Frigg, R. and Werndl, C. [2017]: `Mind the Gap: Boltzmannian versus Gibbsian Equilibrium',  \textit{Philosophy of Science}, \textbf{84:5}, pp. 1289--1302. \vspace{0.5cm} 

\noindent Feynman, R. [1965]: \textit{The Character of Physical Law}, Cambridge: M.I.T. Press.\vspace{0.5cm} 

\noindent Goldstein, S. [2001]: `Boltzmann's Approach to Statistical Mechanics', in J. Bricmont, D. D\"urr, M. C. Galavotti, G. Ghirardi, F. Petruccione, and N. Zanghi (eds), \textit{Chance in Physics. Foundations and Perspectives}, Berlin: Springer, pp. 39--54.  \vspace{0.5cm} 

\noindent  Lavis, D. [2005]: `Boltzmann and Gibbs: An Attempted Reconciliation', \textit{Studies in History and Philosophy of Modern Physics}, \textbf{36}, pp. 245--273.\vspace{0.5cm} 

\noindent  Lavis, D. [2008]: `Boltzmann, Gibbs, and the Concept of Equilibrium', \textit{Philosophy of Science}, \textbf{75:5}, pp. 682--96.\vspace{0.5cm} 

\noindent  Lavis, D. [2011]: `An Objectivist Account of Probabilities in Physics', in C. Beisbart and S. Hartmann (eds), \textit{Probabilities in Physics},  New York: Oxford University Press, pp. 51--82.\vspace{0.5cm} 

\noindent  Lazarovici, D. and Reichert, P. [2015]: `Typicality, Irreversibility and the Status of Macroscopic Laws', \textit{Erkenntnis}, \textbf{80:4}, pp. 689--716. \vspace{0.5cm} 

\noindent  Lazarovici, D. [2018]: `On Boltzmann versus Gibbs and the Equilibrium in Statistical Mechanics', \textit{Philosophy of Science}, \textbf{86:4}, pp. 785--793. \vspace{0.5cm} 

\noindent Lebowitz, J. [1993]: `Macroscopic Laws, Microscopic Dynamics, Time's Arrow and Boltzmann's Entropy', \textit{Physica A}, \textbf{194}, pp. 1--27. \vspace{0.5cm} 

\noindent  Penrose, R. [1989]: \textit{The Emperor's New Mind}, Oxford: Oxford University Press. \vspace{0.5cm} 

\noindent  Penrose, R. [2004]: \textit{The Road to Reality}, London: Vintage. \vspace{0.5cm} 

\noindent  Petersen, K. [1983]: \textit{Ergodic Theory. Cambridge Studies in Advanced Mathematics 2}, Cambridge: Cambridge University Press. \vspace{0.5cm} 

\noindent Reiss, H. [1996]: \textit{Methods of Thermodynamics}, Minneaola/New York: Dover.\vspace{0.5cm} 

\noindent Vranas, P. [1998]: `Epsilon-Ergodicity and the Success of Equilibrium Statistical Mechanics', \textit{Philosophy of Science}, \textbf{65: 4}, pp. 688--708.\vspace{0.5cm} 

\noindent Wilhelm, Isaac. [2019]: `Typical: A Theory of Typicality and Typicality Explanations', \textit{The British Journal for the Philosophy of Science}, axz016, available at \\
<doi.org/10.1093/bjps/axz016>.



\end{document}